\documentclass{article}

\usepackage{amsmath,amssymb,amsfonts,amsthm}
\usepackage{graphicx,color,xcolor,psfrag}
\usepackage{subfigure}
\usepackage{hyperref}
\graphicspath{{./Figures/}}

\newcommand{\setdef}[2]{\left\{#1 \, : \; #2\right\}}
\newcommand{\map}[3]{#1: #2 \rightarrow #3}

\newcommand{\real}{\mathbb{R}}

\newcommand{\realnonnegative}{\mathbb{R}_{\geq0}}
\renewcommand{\natural}{\mathbb{N}}
\newcommand{\integer}{\mathbb{Z}}

\newtheorem{example}{Example}
\newtheorem{remark}{Remark}
\newtheorem{theorem}{Theorem}
\newtheorem{lemma}[theorem]{Lemma}
\newtheorem{proposition}[theorem]{Proposition}
\newcommand{\be}{\begin{equation}}
\newcommand{\ee}{\end{equation}}
\newcommand{\ba}{\begin{array}}
\newcommand{\ea}{\end{array}}

\newcommand{\K}{\mathcal{K}} 	
\newcommand{\I}{\mathcal{I}}		
\newcommand{\co}{\overline{\textup{co}}} 
\newcommand{\diag}{\textup{diag}}		

\newcommand{\Q}{{\bf q}}
\newcommand{\q}{{{\bf \hat q}}}
\newcommand{\1}{{\bf 1}}
\newcommand{\dist}{{\rm dist}}

\newcommand{\dd}{{\rm d}}

\title{Consensus and disagreement: the role of quantized behaviours in opinion dynamics}

\author{Francesca Ceragioli\thanks{Dipartimento di Scienze Matematiche, Politecnico di Torino, Torino, Italy. \texttt{francesca.ceragioli@polito.it}} \and Paolo Frasca\thanks{Department of Applied Mathematics, University of Twente, Enschede, Netherlands. \texttt{p.frasca@utwente.nl}}}

\begin{document}
\maketitle

\begin{abstract}
This paper deals with continuous-time opinion dynamics that feature the interplay of continuous opinions and discrete behaviours. In our model, the opinion of one individual is only influenced by the behaviours of fellow individuals. The key technical difficulty in the study of these dynamics is that the right-hand sides of the equations are discontinuous and thus their solutions must be intended in some generalized sense: in our analysis, we consider both Carath\'eodory and Krasowskii solutions. We first prove existence and completeness of Carath\'e\-o\-dory solutions from every initial condition and we highlight a pathological behavior of Carath\'eo\-do\-ry solutions, which can converge to points that are not (Carath\'eodory) equilibria. Notably, such points can be arbitrarily far from consensus and indeed simulations show that convergence to non-consensus configurations is very common. In order to cope with these pathological attractors, we then study Krasowskii solutions. We give an estimate of the asymptotic distance of all Krasowskii solutions from consensus and we prove its tightness via an example: this estimate is quadratic in the number of agents, implying that quantization can drastically destroy consensus. However, we are able to prove convergence to consensus in some special cases, namely when the communication among the individuals is described by either a complete or a complete bipartite graph.
\end{abstract}


\section{Introduction}
A fundamental assumption in opinion dynamics is that one's opinion is attracted by other's opinions (\cite{NEF-ECJ:11,CC-SF-VL:09}). 
If the opinion of individual $i$ is described  by  a variable $x_i\in \real$, 
then this assumption leads to describe the evolution of the opinions  by a set of differential equations:
 \be
 \label{eq:nonq}\dot x _i(t)=\sum _{j} a_{ij}[x_j(t)-x_i(t)] \qquad 
 i\in{\cal I}=\{ 1,...,N\}
 \ee
 where $a_{ij}\in \realnonnegative $ and $a_{ij}>0$ means that the opinion of the  individual $i$ is influenced by the opinion of the individual  $j$.
 As long as there is at least one individual who can influence (albeit indirectly) all others, then it can be proved that {\em consensus} is asymptotically achieved, namely there exists $\alpha\in \real$ such that $x_i(t)\to \alpha$ as $t\to +\infty$ for all $i\in \cal I$. 
However, experience suggests that in reality consensus is not always achieved, but disagreement persists: clearly, the opinion dynamics model should include some additional feature.

In this paper, we postulate that the individuals can not directly perceive the private opinions of the others, but only observe their displayed {\em behaviours}. Indeed, in some situations individuals may not be able to express their opinions precisely, but only through a limited number of behaviours or actions: let us think of consumers' choices, electoral votes, or stereotyped interactions in online social media.
%
%
%
%
Even though we are assuming that the opinions are real-valued, behaviours are better described as elements of a finite or discrete set:
hence, the behaviour of an individual shall be a suitable {\em quantization} of his/her opinion. 
If we choose for simplicity to represent behaviours as integers, equations~\eqref{eq:nonq} can be replaced by   
 \be\label{eq:dynamics}
 \dot x_i(t)=\sum_{j}a_{ij} [q(x_j(t))-x_i(t)], \qquad i\in\I, 
 \ee
 where the map $\map{q}{\real}{\integer}$ is defined as $q(s)=\lfloor s+\frac{1}{2} \rfloor$. 
This new model is a very simple modification  of~\eqref{eq:nonq}, which has no pretension to describe social dynamics precisely, but has the aim to make the role of quantization evident. The feature that distinguishes it from~\eqref{eq:nonq} is that consensus is not achieved in general. From the point of view of opinion dynamics, this observation allows us to explain the persistence of disagreement as an effect of the limited number of behaviours allowed. 

\begin{figure}[htb]\centering\hspace{-2cm}
\setlength{\unitlength}{0.05in} 
\scalebox{1.5}{
\begin{picture}(40,30)(-20,-15)

{\color{lightgray}
\multiput(0,-8)(0,1){18}{\circle*{0.01}}
\multiput(4,-8)(0,1){18}{\circle*{0.01}}
\multiput(8,-8)(0,1){18}{\circle*{0.01}}
\multiput(12,-8)(0,1){18}{\circle*{0.01}}
\multiput(16,-8)(0,1){18}{\circle*{0.01}}
\multiput(20,-8)(0,1){18}{\circle*{0.01}}
\multiput(0,-8)(1,0){21}{\circle*{0.01}}
\multiput(0,-4)(1,0){21}{\circle*{0.01}}
\multiput(0,4)(1,0){21}{\circle*{0.01}}
\multiput(0,8)(1,0){21}{\circle*{0.01}}
}

\put(10,-12){\vector(0,1){24}}
\put(-2,0){\vector(1,0){24}}
\thicklines
\put(8,0){{\line(1,0){4}}}
\put(12,4){{\line(1,0){4}}}
\put(4,-4){{\line(1,0){4}}}
\put(0,-8){{\line(1,0){4}}}
\put(16,8){{\line(1,0){4}}}
\put(8,0){\circle*{1}}
\put(12,4){\circle*{1}}
\put(4,-4){\circle*{1}}
\put(0,-8){\circle*{1}}
\put(16,8){\circle*{1}}
\put(20,-3){${ s}$}
\put(4,10){${ q(s)}$}

\put(5,-2){{\tiny $-1/2$}}
\put(11,-2){{\tiny $1/2$}}
\put(8,4){{\tiny $1$}}
\put(7,-5){{\tiny $-1$}}
\end{picture}
\vspace{-5mm}
}
\caption{The uniform quantizer $\map{q}{\real}{\integer}$.}\label{fig:quantizer}
\end{figure}

\smallskip
The aim of this paper is to study the effect of quantization of others' states in a consensus model. This is done by undertaking a rigorous mathematical analysis of~\eqref{eq:dynamics}.
The key technical difficulty in the study of this model is that the right-hand sides of the equations are discontinuous and their solutions must be intended in some generalized sense. As the {\em first} contribution, we prove existence and completeness of Carath\'eodory solutions from every initial condition. This is a relevant fact because most discontinuous systems, including well-known models of opinion dynamics, do not have complete Carath\'eodory solutions~\cite{FC-PF:11}: even proving existence for most initial conditions can be a daunting task~\cite{VDB-JMH-JNT:09a,VDB-JMH-JNT:09b}.  

As the {\em second} contribution, we highlight a pathological behavior of Carath\'eo\-do\-ry solutions, which can converge to points that are not (Carath\'eodory) equilibria: actually, these points can be arbitrarily far from consensus. The presence of these pathological attractors leads us to consider also Krasowskii solutions, which are trivially seen to exist and be complete. Indeed, Krasowskii solutions can not converge to points that are not (Krasowskii) equilibria. Another advantage is that, as Carath\'eodory solutions are particular Krasowskii solutions, results that hold for Krasowskii solutions are more general. 

As the {\em third} contribution, we derive an estimate of the asymptotic distance of Krasowskii solutions from consensus and we prove by means of an example that it can not be substantially improved in general. This estimate is quadratically increasing in the number of agents and its tightness implies that quantization can drastically destroy consensus.  Actually, simulations show that most graphs imply convergence to non-consensus configurations. However, consensus can happen in some special cases.
Indeed, as the {\em fourth} contribution, we prove convergence to consensus when communication among the individuals is described either by a complete graph or by a complete bipartite graph.

\smallskip
\paragraph{Relations with the literature.}
This paper relates to various bodies of work in control theory and in mathematical sociology.
In the last ten years, control theorist have widely studied quantized versions of consensus algorithms. Since giving a complete overview would be impossible here, we just mention some papers whose approaches are particularly close to ours.
First of all, in~\cite{RC-FF-PF-TT-SZ:07} the authors consider a discrete-time version of~\eqref{eq:dynamics}: their analysis is limited to observing that the algorithm may not converge to consensus and then abandoned. The poor perfomance of the dynamics in terms of approaching consensus explains the absence\footnote{A preliminary and incomplete study on the topics of this article was published in the Proceedings of the European Control Conference as~\cite{FC-PF:15}, where only all-to-all communication was considered.} of known results about~\eqref{eq:dynamics}. 
Instead, papers like~\cite{FC-CDP-PF:10a,PF:11,DVD-KHJ:10,MG-DVD:13,SL-TL-LX-MF-JZ:13} have considered other possible quantizations of~\eqref{eq:nonq}: in~\cite{FC-CDP-PF:10a,PF:11} all states are quantized in the right-hand side, while in~\cite{DVD-KHJ:10,MG-DVD:13} distances between couples of states are seen through the quantizer. In these papers convergence to consensus is proved under appropriate but generally mild assumptions~\cite{JW-XY-HS-KHJ:16}. 

%
In the opinion dynamics literature, different explanations of persistence of disagreement have been proposed: among them, we recall in~\cite{NEF-ECJ:11} the persisting influence of the initial opinions; in~\cite{RH-UK:02} the bounded confidence of the individuals; in~\cite{MM-AP-SR:07,DA-GC-FF-AO:11} the presence of stubborn individuals; and in~\cite{CA:13} the occurrence of antagonistic interactions.
Here, we support the claim that quantization can be a source of disagreement. The fact that others' opinion are materialized by means of discrete behaviours is a natural observation, which has been made by social scientists and socio-physicists and has been addressed in a few models including~\cite{DU:03,ACRM:08,NC-IM-SM-SS:16} and~\cite[Chapter~10]{NEF-ECJ:11}. Discrete behaviours are the outcome of limited verbalization capabilities in~\cite{DU:03} or represent actions taken by the indiduals (CODA models) in~\cite{ACRM:08}. 
These papers feature dynamics that may not reach consensus and that involve quantization together with other, possibly nonlinear, effects. In comparison, our proposed model~\eqref{eq:dynamics} can be understood as an effort to single out the effects of quantization only. We consider the case of multiple quantization levels, but two-level quantization as in~\cite{NC-IM-SM-SS:16} can be easily obtained as a particular case.


\paragraph{Outline of the paper.}
In Section~\ref{sec:preliminary} we present some alternative forms of~\eqref{eq:dynamics}, introduce Carath\'eodory and Krasowskii solutions and prove some of their basic properties.
Section~\ref{sec:equilibria} is devoted to equilibria, providing the relevant definitions and examples.
In Section~\ref{sec:quant-dist} we prove the asymptotic estimate of the distance of solutions from consensus and show that it is sharp.
In Section~\ref{sec:consensus} we study the special cases of all-to-all and bipartite all-to-all communication and prove that in these cases consensus is achieved. Finally, Section~\ref{sec:simulations} presents some simulations and Section~\ref{sec:outro} concludes the paper.

\section{Fundamental properties of the dynamics}\label{sec:preliminary}

We begin this section by rewriting equations~\eqref{eq:dynamics} in some alternative forms that allow us to see the model from different points of view. 

First of all, we can  think that communication among individuals is described by a weighted directed graph whose vertices are individuals. The numbers 
$a_{ij}\ge0$ are the entries of the weighted adjaciency matrix $A$ of the graph, namely $a_{ij}>0$ if the agents $i$ and $j$ are linked, $a_{ij}=0$ if they are not linked and $a_{ii}=0$.
We can then observe that 
\begin{align*}
\dot x _i(t)=& \sum_{j\not =i}a_{ij} [q(x_j(t))-x_j(t)+x_j(t)
  -x_i(t)]\\
=& \sum_{j\not =i}a_{ij} (x_j(t))-x_i(t))+\sum_{j\not =i}a_{ij}
 (q(x_j(t))-x_j(t))
\end{align*}
If we denote by $d_i:=\sum_{j}a_{ij}$ the weighted degree of the $i$-th vertex, by $D=\diag (d_i)$, and by $\Q (x)=(q(x_1),\ldots ,q(x_N))^\top $, 
then we can define the Laplacian matrix of the graph $L=D-A$ and write
\begin{align}\label{eq:dynamics-L}
\dot x=& -Lx+A(\Q (x)-x)\\\nonumber=&-L(x-x_a\1)+A(\Q (x)-x)
\end{align}
where for the second equality we have used $\1 ^\top =(1,...,1)$ and
 $x_a(t)=\frac{1}{N}\1^\top x(t)$, together with the fact that $L\1=0$. System~\eqref{eq:dynamics-L} can be seen as the classical consensus system perturbed by others' states quantization errors. This interpretation will become useful in Section~\ref{sec:quant-dist}.
Clearly, we can also write~\eqref{eq:dynamics-L} as
\be\label{eq:dynamics-nn}
\dot x=-Dx+A\Q (x),
\ee
which prompts a connection with the theory of neural networks.
\begin{remark}[Neural networks] By adding a constant term $I\in \real^N$ in~\eqref{eq:dynamics-nn} we get 
\be
\dot x=-Dx+A\Q (x)+I,
\ee
which is the typical system used in order to describe neural networks. In this context, individuals represent neurons and the function $q$ is called activation function. In the basic models $q$ is assumed to be smooth and increasing, but the literature has also considered weaker assumptions. In particular, in~\cite{MF-PN:03} $q$ can be discontinuous but it must be increasing and bounded,  and the matrix $A$ is assumed invertible.  A major difference between our model and the one in~\cite{MF-PN:03} is that we have multiple equilibria (Section~\ref{sec:equilibria}), whereas \cite{MF-PN:03} assumes a single equilibrium point. 
\end{remark}
%

When we look at system~\eqref{eq:dynamics-nn},
we observe that for every ${\bf k}\in \integer^N$ the vector $\Q$ is constant on each set 
$$S_{{\bf k}}=\{ x\in \real ^N: k_i-\frac{1}{2}\leq x_i <k_i+\frac{1}{2}, i=1,\dots,N \}.$$ 
This fact makes it evident that the system is affine if restricted to each set $S_{{\bf k}}$. Moreover its right-hand side is discontinuous on the set $\Delta=\bigcup_{{\bf k} \in \integer^N} \partial S_{\bf k}$, where $\partial S_{\bf k}$ is the boundary of $S_{\bf k}$.  In general, existence of solutions of equations with discontinuous right-hand side is not guaranteed. 
For this reason different types of solutions have been introduced in the literature (see~\cite{JC:08-csm}). 
The notion of solution nearest to the classical one is that of Carath\'eodory solution. The main problem in using Carath\'eodory solutions is that, often, they do  not exist.   Here we prove that Carath\'eodory solutions do exists and are complete.  Nevertheless the discontinuity causes the strange phenomenon of solutions converging to  points which are not equilibria.  In order to better understand this phenomenon we also study Krasowskii solutions: such points are in fact Krasowskii equilibria.  

We now formally introduce Carath\'eodory and Krasowskii solutions and study their basic properties. We refer the reader to~\cite{JC:08-csm,OH:79} for overviews on generalized solutions of discontinuous differential equations. System~\eqref{eq:dynamics} can be cast in the general form 
\be
\label{eq:dynamics-f}
\dot x= f(x),
\ee
provided $\map{f}{\real ^N}{\real ^N}$ is the vector field given by $f_i(x)=\sum_{j}a_{ij}[q(x_j)-x_i]$.
Let $I\subset \real$ be an interval of the form $(0,T)$. An absolutely continuous function $\map{x}{I}{\real^N}$ is a {\em Carath\'eodory solution} of~\eqref{eq:dynamics-f} if it satisfies~\eqref{eq:dynamics-f} for almost all $t\in I$ or, equivalently, if it is a solution of the integral equation
$$x(t)=x_{0}+\int _0^t f(x(s))\dd s.$$
%
An absolutely continuous function $\map{x}{I}{\real^N}$ is a Krasowskii solution of~\eqref{eq:dynamics-f} if
for almost every $t\in I$,  it satisfies
\be\label{eq:Krasowskii}
\dot x(t)\in \K f(x(t)),
\ee
where $$ \K f(x) =\bigcap_{\delta>0}\co(\setdef{f(y)}{y \text{ such that } \|x-y\|<\delta}).$$
%
We recall that in general any Carath\'eodory solution is also a Krasowskii solution and we observe that, for the specific dynamics~\eqref{eq:dynamics-f}, Krasowskii solutions coincide with Filippov solutions that are often adopted for discontinuous systems.
The following result states the basic properties of the solutions of~\eqref{eq:dynamics-f}.

\begin{theorem}[Properties of solutions]\label{prop:basic_properties}
\begin{itemize}

\item [(i)] (Existence) For any initial condition there exist a Carath\'eodory solution and a Krasovskii solution of~\eqref{eq:dynamics}.
\item [(ii)] (Boundedness) Any Carath\'eodory or Krasowskii solution of~\eqref{eq:dynamics} is bounded on its domain.

\item [(iii)] (Completeness)  Any Carath\'eodory or Krasowskii  solution starting at
  $t_0\in \real$ is defined on the set $[t_0, +\infty )$.

\end{itemize}
\end{theorem}

\begin{proof}  
(i-C)
First of all, we remark that the right-hand side of~\eqref{eq:dynamics} is continuous at any point in the interior of 
$S_{{\bf k}}$ for any ${\bf k}\in \integer ^N$, then 
 local  solutions with initial conditions in $\real ^N\setminus\Delta$ do exist.
Then, we consider initial conditions in $\Delta$.
   For any $x_0\in \real ^N$ we denote by $I(x_0)$ the subset of
$\{1,\ldots ,N \}$ of the indices $i$ such that 
$x_{0i}=k_i+\frac{1}{2}$ for some
$k_i\in \integer $ and by $M$ the cardinality of $I(x_0)$.

We first consider initial conditions $x_0$ such that $M=1$ and $I(x_0)=\{ i\}$,
i.e.\  
$x_{0i}=k_i+\frac{1}{2}$ for some $k_i\in \integer $ and
$x_{0j}\not =h+\frac{1}{2}$ for any $j\not =i$ and any  $h\in \integer $. 
Let us denote
 $$s_{i}(x)=x_i-k_i-\frac{1}{2},$$ 
$$S_{i}^+=\{ x\in \real ^N: x_i-k_i-\frac{1}{2}\geq 0\},$$  
$$S_{i}^-=\{ x\in \real ^N: x_i-k_i-\frac{1}{2}<0\},  $$ 
$$f| _{S_{i}^-}(x_0)=\lim _{x\in S_{i}^-, x\to x_0 }f(x).$$
We have that $\nabla s_{i}(x_0)=e_i$ (the $i$-th vector of the canonical basis)
$$
\ba{rl}
a(x_0)& =\nabla s_{i}(x_0)\cdot f| _{S_{i}^+}(x_0)= \nabla s_{i}(x_0)\cdot f| _{S_{i}^-}(x_0)\\
& =\sum_{j\not =i}a_{ij}(q(x_{0j})-k_i-\frac{1}{2}).
\ea $$
 If $a(x_0)<0$, then there is a solution starting at $x_0$ which satisfies
the equations $\dot x= f| _{S_{i}^-}(x)$ and stays in $S_{i}^-$ in an interval of
the form $(t_0, t_0+\epsilon )$ for some $\epsilon >0$.
If $a(x_0)\geq 0$, then there is a solution starting at $x_0$ which satisfies
the equation $\dot x= f| _{S_{i}^+}(x)$ and stays in $S_{i}^+$ in an interval of
the form $(t_0, t_0+\epsilon )$ for some $\epsilon >0$.
Note in particular that if $a(x_0)=0$, then the vector field $f|
_{S_{i}^+}$ is tangent to $S_{i}^+$ not only in $x_0$ but also in a neighbourhood of $x_0$ intersecated with the discontinuity surface. 
In fact there exists a neighbourhood $N(x_0)$ of $x_0$ such that for all $x\in N(x_0)$ one has $q(x_j)=q({x_0}_j)$ for all $j\not =i$. Moreover if $x_i=k_i+1/2$ one gets $a(x)=a(x_0)$. 

We now consider initial conditions $x_0$ such that 
  $1<M\leq N$. The
vector field $f$ has $2^M$ limit values at $x_0$ corresponding to the
$2^M$ sectors defined by the inequalities $x_i-k_i-\frac{1}{2}\geq 0$
and   $x_i-k_i-\frac{1}{2}< 0$, $i\in I(x_0)$. 
We describe  these sectors by means of
 vectors $B\in \{ 0,1\}^N\subset \real ^N$. Let $H(t)=1$ if $t\geq 0$ and $H(t)=0$ if
$t<0$.
We  define $B_i= H(x_i-k_i-\frac{1}{2})$ if
$i\in I(x_0)$ and $B_i=0$ if $i\notin I(x_0)$.
Let $f_B(x_0)=\lim _{x\to x_0, x\in B}f(x)$. In the following we denote $f_B(x_0)=f_B$. 
We want to prove that there exists  $B$ such that
$H((f_B)_i)=B_i$ for all $i\in I(x_0)$. This means that the vector
field $f_B$ points inside the sector $B$ at $x_0$.

Preliminarily, note that the $i$th component of $f_B$ can be written as 
\be\label{eq:f_B}
\ba{rl}
(f_B)_i& =\sum_j a_{ij}(k_j+B_j)-d_i (k_i+\frac{1}{2})= \sum_j a_{ij}k_j-d_i (k_i+\frac{1}{2})+\sum_{j}a_{ij}B_j
\ea
\ee
Now, we start by considering the sector $B^1$ such that $(B^1)_i=0$ for all
$i\in I(x_0)$. If  $H((f_{B^1})_i)=0$ for all $i\in I(x_0)$, we have finished. Otherwise, there
exists $i\in I(x_0)$ such that $H((f_{B^1})_i)=1$. Assume without loss of generality that $i=1$, i.e.\ $H((f_{B^1})_1)=1$. Then for all $B\not = B^1$ we have
$H((f_{B})_1)=1$. 
In fact if $(f_{B^1})_1>0$ then also $(f_B)_1=(f_{B^1})_1+\sum_{j}a_{1j}B_j>0$.

We then  examine only those $B$ such that
$B_1=1$. In particular the next $B$ we consider, which we call $B^2$,
is such that $(B^2)_1=1$ and $(B^2)_i=0$ for all other $i\in
I(x_0)$. If  $H((f_{B^2})_i)=0$ for all $i\in I(x_0)\setminus \{ 1\}$, then we have
finished. Otherwise, there exists $j\in I(x_0)\setminus \{ 1\}$ such that $H((f_{B^1})_j)=1$. Assume that such $j=2$, i.e.\ $H((f_{B^2})_2)=1$.  Then for all $B\not\in\{ B^1, B^2\}$ we have
$H((f_{B})_2)=1$. We can then restrict our attention to those $B$ such
that $B_1=B_2=1$, and so forth. By proceeding in this way, in $M$ step at most we
find the sector $B$ with the desired property. 
 
As already mentioned, the meaning of the condition $H((f_B)_i)=B_i$ for all $i$ is that the
vector field $f_B$ is directed inside $B$. 
If $f_B$ points strictly inside $B$, i.e.\  $(f_B)_i\not =0$ for all $i$, then there is a solution that enters the sector.
If instead some components of $f_B$ are null, $f(B)$ is tangent to the boundary of $B$, and, more precisely, to 
$\partial B\cap [\cap_{i\in \tilde{I}(x_0)}\{ x: x_i=k_i+1/2\}]$, where
$\tilde I(x_0)$ is the subset of $I(x_0)$ such  that  $(f_B)_i=0$ if $i\in \tilde I (x_0)$.
Thanks to~\eqref{eq:f_B} there exists a neighbourhood $N(x_0)$ of $x_0$ such that  $(f_B(x))_i=0$ for all $i\in \tilde I(x_0)$ and  for all $x\in N(x_0)\cap \partial B\cap [\cap_{i\in \tilde{I}(x_0)}\{ x: x_i=k_i+1/2\}]$, i.e.\ $f_B(x)$ remains tangent to $\partial B\cap [\cap_{i\in \tilde{I}(x_0)}\{ x: x_i=k_i+1/2\}]$ so that there exists a solution
$\varphi $ of~\eqref{eq:dynamics}, and $t_0\in \real, \epsilon >0$ 
 such that $\varphi (t_0)=x_0$ and  $\dot \varphi (t)=f_B(\varphi
 (t))$ for almost every $t\in (t_0,t_0+\epsilon )$.

(i-K) Existence of Krasowskii solutions follows from the fact that the function $q$ is locally bounded.

(ii-C) Let $x(t)$ be a Carath\'eodory solution and let  $m$ be any index in $\I$  such that $x_m(t)= \min \{ x_i(t), i\in \I\}$. For $i\in I$, \   $x_i(t)\geq x_m(t)$. 
Let $q(x_m(t))=q_m(t)$. If $x_m(t)\in \left[q_m(t)-\frac{1}{2}, q_m(t) \right]$, one has $q(x_i(t))\geq x_m(t)$ for all $i\in\I$ and 
then $\dot x_m (t)=\sum_{j}a_{ij} [q(x_j(t))-x_m(t)]\geq 0$. This implies that $x_m(t)$ is lower bounded by $\min\{ x_m(0), q_m(0)\}$. 
The previous calculation also shows that the function $q_m(t)$ is nondecreasing.
Analogously, if  $M$ is any index such that $x_M(t)= \max \{ x_i(t), i\in \I\}$, $q(x_M(t))=q_M(0)$, we get that $x_M(t)$ is upper bounded by $\max\{ x_M(0), q_M(0)\}$, and the fuction $q_M(t)$ is nonincreasing.  

(ii-K) Let now $x(t)$ be a Krasowskii solution $m,M, x_m,x_M, q_m,q_M$ be defined as in (ii-C). 
As for Carath\'eodory solutions, if $x_m(t)\in \left(q_m(t)-\frac{1}{2}, q_m(t) \right]$, one has $q(x_i(t))\geq x_m(t)$ for all $i\in\I$ and 
then $\dot x_m (t)=\sum_{j}a_{ij} [q(x_j(t))-x_m(t)]\geq 0$
while if $x_m(t)=q_m(t)-\frac{1}{2}$, $\dot x_m(t)$ may be negative and then 
 $x_m(t)$ is lower bounded by $q_m(0)-1$ and $x_M(t)$ is upper bounded by $q_M(0)+1$.

(iii-C-K) Both Carath\'eodory and Krasowskii solutions can be continued up to $+\infty$ thanks to local existence and boundedness of solutions. 
\end{proof}

\begin{remark}[Monotonicity and limit of minimum and maximum quantization level for Carath\'eodory solutions]\label{rem:monotonicity}
Many consensus-seeking dynamics enjoy the property that the smallest (largest) component is nondecreasing (nonincreasing). This fact is {\em not} true for system~\eqref{eq:dynamics}. Instead, we have shown in the proof of Theorem~\ref{prop:basic_properties} that, for Carath\'eodory solutions, the smallest and the largest quantization levels $q_m(t)$ and $q_M(t)$ are nondecreasing and nonincreasing, respectively. Since they are bounded  and take values in $\integer$, it follows that they are definitively constant. 
For any Carath\'eodory solution $x(t)$ of~\eqref{eq:dynamics} there exist $T^*\in \real$ and  $q^*_m, q^*_M\in \integer$ such that for
any $t\geq T^*$ $\min \{ q(x_i(t)), i=1,\dots, N\}=q^*_m$ and $\max \{ q(x_i(t)), i=1,\ldots ,N\}=q^*_M$. This monotonicity of the extremal quantization levels does {\em not} hold for Krasowskii solutions, as can be seen from Example~\ref{example:4-path-contd} in the next section.
\end{remark}

\begin{remark}[Infinite versus finite quantization levels] 
The function $q$ takes values in $\integer$, and this means that we a priori admit infinite types of behaviours. Nevertheless a consequence of the previous remark is that the number of quantization levels actually assumed in the evolution of the system is finite, once the initial condition is fixed.  
In particular, if the initial values of all the components are in the interval $(0,1)$, the quantizer only takes the  values $0$ and $1$: in this way we get  the case of binary behaviours. 
\end{remark}

The following example shows that in general uniqueness of solutions is not guaranteed. 
\begin{example}[Multiple Carath\'eodory solutions]\label{example:non-unique}
Consider the dynamics 
\begin{align*}
\dot x_1=&q(x_2)-x_1\\
\dot x_2=&q(x_1)-x_2
\end{align*} with initial condition
$(x_1(0), x_2(0))^\top =(1/2,1/2)^\top $. There are two solutions issuing from this
point, whose trajectories are the line segments joining the initial
condition with the points $(0,0)^\top $ and $(1,1)^\top $; see Figure~\ref{fig:non-unique}.
\end{example}

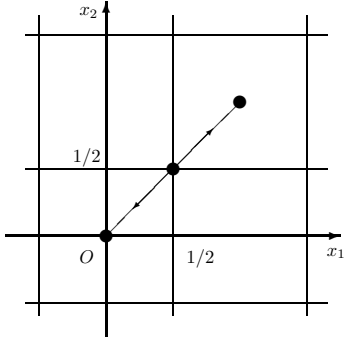
\begin{figure}
\setlength{\unitlength}{0.05in} 
\centerline{ 
\scalebox{0.7}{
\begin{picture}(60,50)(-30,-25)
{\thicklines
\put(-10,-25){\vector(0,1){50}}
\put(-25,-10){\vector(1,0){50}}
}
\put(0,0){\circle*{2}}
\put(-10,-10){{\circle*{2}}}
\put(-14,-14){$O$}
\put(10,10){{\circle*{2}}}
\multiput(-20,-22)(20,0){3}{\line(0,1){45}}
\multiput(-22,-20)(0,20){3}{\line(1,0){45}}
\put(0,0){\line(1,1){10}}
\put(0,0){\line(-1,-1){10}}
\put(0,0){\vector(1,1){6}}
\put(0,0){\vector(-1,-1){6}}
\put(-14,23){$x_2$}
\put(23,-13){$x_1$}
\put(2,-14){$1/2$}
\put(-15,1){$1/2$}
\end{picture}}
}\caption{Non-uniqueness of the solution through the point
    $(1/2, 1/2)$.}
    \label{fig:non-unique}
\end{figure}


\section{Equilibria and lack of consensus}\label{sec:equilibria}


Generally speaking, a point $x^*$ is said to be an equilibrium of~\eqref{eq:dynamics-f} if the function $x(t)=x^*$ is a solution of~\eqref{eq:dynamics-f}.
Since we have different types of solution, we distinguish between different types of equilibria.

\smallskip 
We call
$x^*$  a {\em Carath\'eodory equilibrium} of~\eqref{eq:dynamics-f} if the function $x(t)=x^*$ is a Carath\'edodory solution of~\eqref{eq:dynamics-f}, and we call
$x^*$  a {\em Krasowskii equilibrium} of~\eqref{eq:dynamics-f} if the function $x(t)=x^*$ is a Krasowskii solution of~\eqref{eq:dynamics-f}.
Carath\'eodory equilibria are found by looking for solutions of the equation  $f(x^*)=0$, whereas  Krasowskii equilibria are points $x^*$  such that $0\in Kf(x^*)$.
If we denote by $E_C$ the set of Carath\'eodory equilibria of~\eqref{eq:dynamics-f} and by $E_K$ the set of Krasowskii equilibria, it is evident that $E_C\subset E_K$. As an example of a Krasowskii equilibrium that is not a Carath\'eodory equilibrium we can take the point $(1/2,1/2)^\top $ in Example~\ref{example:non-unique}. Actually points of the form $(h+1/2,...,h+1/2)^\top $ are always Krasowskii equilibria of~\eqref{eq:dynamics}.

\smallskip
Let us call {\em consensus point} a point $\overline x\in \real ^N$ such that $\overline x _i=\overline x_j$ for all $i,j=1,...,N$.
Clearly, integer consensus points like $x^*=k\1$ with $k\in \integer$ are Carath\'eodory equilibria  of~\eqref{eq:dynamics} because $f(k\1)=0$. Such a point $k\1$ belongs to the interior of $S_{k\1}$ and, consequently, is locally asymptotically stable.
However, Carath\'eodory equilibria are not necessarily consensus points and may belong to the discontinuity set.
\begin{example}[Non-consensus Carath\'eodory equilibrium] Consider the system
\begin{align}
\nonumber\dot x_1 =& q(x_3)-x_1\\
\nonumber\dot x_2 =& q(x_3)-x_2 \\
\nonumber\dot x_3=& q(x_1)+q(x_2)+q(x_4)-3x_3\\
\nonumber\dot x_4 =& q(x_3)+q(x_5)-2x_4\\
\nonumber\dot x_5 =& q(x_4)-x_5
\end{align}
The point $\bar x=(0,0,1/3,1/2,1)^\top $ is a Carath\'eodory equilibrium point which lies on the boundary of $S_{(0,0,0,1,1)}$. The set $S_{(0,0,0,1,1)}$ is contained in the attraction region of $\bar x$, but $\bar x$ is not Lyapunov stable. 
\end{example}
\noindent This example also shows that solutions {\em can converge to points that are not consensus points}.

For smooth systems,  if a (classical) solution converges to a point, then such point is a (classical) equilibrium. 
This property does not hold true for Carath\'eodory solutions of  systems with discontinuous right-hand side, and, in particular, in the case of systems of the form~\eqref{eq:dynamics-nn}. This motivates  the following definition of extended equilibrium.
Let ${\bf k}\in\integer$ and $f_{\bf k}$ the vector field whose components are
$$(f_{\bf k})_i(x)=\sum_{j\not =i}a_{ij}(k_j-x_i).$$ Clearly $f_{\bf k}$
coincides with $f$ on the set $S_{\bf k}$.
We call {\em extended equilibrium} of~\eqref{eq:dynamics}  a point $x^*\in \real ^N$ such that there exists ${\bf k}\in
\integer ^N$ such that $f_{\bf k}(x^*)=0$ and $x^*\in \overline{S_{\bf k}}$. We denote by $E_e$ the set of  extended equilibria of~\eqref{eq:dynamics}. 
It is evident that $E_C\subset E_e \subset E_K$.
The following example shows that these inclusions are strict in general.
\begin{example}[4-path graph: equilibria]\label{example:4-path}
Let us consider~\eqref{eq:dynamics} on a 4-node line graph:
\begin{align}\label{eq:4line}
\nonumber\dot x_1 =& \,q(x_2)-x_1\\
\dot x_2 =&\, q(x_1)+q(x_3)-2x_2 \\
\nonumber\dot x_3=& \,q(x_2)+q(x_4)-2x_3\\
\nonumber\dot x_4 =& \,q(x_3)-x_4
\end{align}
Then,
\begin{itemize}
\item
The point $x^A=(0,\frac12,\frac12,1)^\top $ is an extended equilibrium, because $x^A \in \overline{S_{(0,0,1,1)}}$ and $f _{(0,0,1,1)}(x^A)=0$. 
However, $x^A$ can not be a Carath\'eodory equilibrium, because $q(x^A_2)=1\neq x^A_1=0$.
\item The point $x^B=(\frac12,\frac12,\frac12,\frac12)^\top $ is a Krasowskii equilibrium, as it can be computed $0\in \K f(x^B)$. However, $x^B$ can not be an extended equilibrium, because $x^B_1=\frac12$ can not be equal to any quantizer value, so that the first component of the vector field can not be zero in any neighbourhood of the point.
\end{itemize}
\end{example}
This example shows that extended equilibria include points, like $x^A$, that are not consensus points. Interestingly, there exist Carath\'eodory solutions that asymptotically converge to $x^A$: it is enough to take the solution issuing from an initial condition in $S_{(0,0,1,1)}$.
In spite of being attractive, $x^A$ is {\em not} a Carath\'eodory equilibrium and actually Carath\'eodory solutions originating from $x^A$ converge to $(1,1,1,1)^\top $.
This pathological behaviour might appear surprising, but is allowed by the discontinuity of the vector field.


In the previous example one can find a Krasowskii solution that is not a Carath\'eodory solution: its slides on the discontinuity set and connects two equilibria. 
\begin{example}[4-path graph: Krasowskii trajectories]\label{example:4-path-contd}
Let us consider again the dynamics~\eqref{eq:4line} on the 4-node path graph.
Consider the parametrized segment $x_a=(a,\frac12,\frac12,1-a)^\top $ for $a\in [0,\frac12]$, which interpolates between the Krasowskii equilibria  $x_0=x^A$ and $x_{\frac12}=x^B$. For every $a\in (0,\frac12)$ it holds that 
$$ \K f(x_a)=\overline{\rm co} \left\{\right (1-a,0,1,a)^\top , (-a,0,0,a)^\top ,(-a,-1,0,-1+a)^\top , (1-a,-1,1,-1+a)^\top   \}.$$
Then there is a Krasowskii solution $\varphi (t)$ such that 
\begin{itemize}
\item{}  $\varphi(0)= x_B=(1/2,1/2,1/2,1/2)$, 
\item{}  for any $t>0$ there exists $a\in (0,1/2)$ such that 
 $\varphi (t)=x_a$, 
 \item{} if $\varphi (t)=x_a$ then   $\dot \varphi (t)=(-a,0,0,a)$
 \item{}  $\lim_{t\to +\infty}\varphi (t)=x_A=(0,1/2,1/2,1)$.
  \end{itemize}
 Note that $\varphi (t)$ can not be a Carath\'eodory solution as $f(x_a)= (1-a,0,1,a)$.
\end{example}
This example also shows that the minimum quantizer level $q_m(t)$ assumed along Krasowskii solutions may be decreasing (see Remark~\ref{rem:monotonicity}). In fact, if we compute $q_m(t)$ for the Krasowskii solution showed above, then  
$q_m(0)=1$ whereas $q_m(t)=0$ if $t>0$.

\subsection{Path graph: equilibria}
We now show, by means of an example, that equilibria can be significantly far from consensus: in the case the communication graph is a path, such distance can be proportional to the square of number of individuals.

In the case the graph is a path, equations~\eqref{eq:dynamics} read
\begin{align}\label{eq:line}
\dot x_1 =& q(x_2)-x_1 \nonumber\\
\dot x_i =& q(x_{i-1})+q(x_{i+1})-2x_i, \  i=2,\ldots ,N-1 \\
\dot x_N =& q(x_{N-1})-x_N \nonumber
\end{align}

\begin{proposition}\label{prop:q-equil-path}
If $x^* $ is an extended equilibrium, then
$$\max\setdef{ |x^*_j-x^*_i|}{ i,j=1,\ldots,N}\leq \frac{(N-2)^2}{4}.$$
Furthermore, there exists an extended equilibrium $x^* $ such that 
$$x^*_N-x^*_1=\begin{cases}
\frac{(N-2)^2}{4} &\text{if $N$ is even}\\
\frac{(N-1)(N-3)}{4} &\text{if $N$ is odd.}
\end{cases}$$

\end{proposition}

\begin{proof}
Let us characterize extended equilibria of~\eqref{eq:line}.
Let ${\bf k}\in \integer^N$. If $x^*$ is such that $f_{\bf k}(x^*)=0$
then 
\begin{align*}
x^*_1= & k_2\\
x^*_i=& \frac{k_{i-1}+k_{i+1}}{2},\  i=2,\ldots ,N-1\\
x^*_N= & k_{N-1}.
\end{align*}
Then, $x^*\in \overline S_{\bf k}$ if and only if the following inequalities
are satisfied 
\begin{align*}
k_1-\frac{1}{2}\leq &  k_2\leq k_1-\frac{1}{2}\\
k_i-\frac{1}{2}\leq & \frac{k_{i-1}+k_{i+1}}{2}\leq k_i+\frac{1}{2},\ i\in\{ 2,\ldots ,N \}\\
k_N-\frac{1}{2}\leq & k_{N-1}\leq k_N+\frac{1}{2}
\end{align*}
Since $k_i\in \integer$,  from the first inequality it follows that
$k_1=k_2$ and from the last one it follows that $k_N=k_{N-1}$. 
Moreover the inequalities  with $i=2,\ldots , N-1$ can be written in the
 form
$$k_i-k_{i-1}-1\leq k_{i+1}-k_i\leq k_i-k_{i-1}+1$$ 
that implies
$$|k_{i+1}-k_i|-|k_{i}-k_{i-1}|\leq |(k_{i+1}-k_i)-(k_{i}-k_{i-1})|\leq 1$$ 
i.e.
$$|k_{i+1}-k_i|\leq |k_i-k_{i-1}|+1 \ \text{and} \   |k_{i}-k_{i-1}|\geq |k_{i+1}-k_{i}|-1. $$
From the left inequality  we deduce that 
$$|k_i-k_{i-1}|\leq i-2,\  i\in\{ 2,\ldots ,N \}$$
and from the right  inequality
$$|k_i-k_{i-1}|\leq N-i,\  i\in\{ 2,\ldots ,N \}.$$
These two inequalities imply that
$$|k_i-k_{i-1}|\leq \min \{ i-2, N-i\},\   i\in\{ 2,\ldots ,N \},$$
that can be written as
$$|k_i-k_{i-1}|\leq  i-2    , \   i\in\{ 2,\ldots ,N\}, i\leq \frac{N+2}{2}$$
and 
$$|k_i-k_{i-1}|\leq  N-i , \   i\in\{ 2,\ldots ,N \}, i> \frac{N+2}{2}.$$
\smallskip
Assuming without loss of generality that $j\geq i$, we can deduce that 
$$|k_j-k_i|\leq
|k_{j}-k_{j-1}|+|k_{j-1}-k_{j-2}|+\ldots +|k_{i+1}-k_i|\leq |k_N-k_{N-1}|+\ldots +|k_2-k_1|$$
and also 
$$|x_j-x_i|\leq |k_N-k_{N-1}|+\ldots +|k_2-k_1|$$

\noindent
Let us distinguish the cases $N$ is even and $N$ is odd.
If $N$ is even, $N=2n, n\in \natural$, and $\frac{N+2}{2}=n+1$. We have 
\begin{align*}
& |k_N-k_{N-1}|+|k_{N-1}-k_{N-2}|+\ldots+|k_{n+1}-k_{n}|+\ldots+|k_2-k_1|\\
&\leq  \big(0+1+2+\ldots+(2n-n-2)\big)+\big((n+1-2)+(n-2)+\ldots+1+0\big)\\
&\leq \frac{(n-2)(n-1)}{2}+\frac{(n-1)n}{2}=(n-1)^2=\frac{(N-2)^2}{4}.  
\end{align*}
If $N$ is odd, $N=2m+1, m\in \natural$, and $\frac{N+2}{2}=m+\frac{3}{2}$. We have 
\begin{align*}
& |k_N-k_{N-1}|+|k_{N-1}-k_{N-2}|+\ldots |k_{n+1}-k_{n}|+\ldots .+|k_2-k_1|\\
& \leq   \big(0+1+2+\ldots+(2m+1-m-2)\big)+\big(m+1-2)+(m-2)+\ldots+1+0\big)\\
& \leq  2\frac{(m-1)m}{2}=(m-1)m=\frac{(N-1)(N-3)}{4}.  
\end{align*}
The first statement then follows from the fact $(N-1)(N-3)\leq (N-2)^2$.

On the other hand, if we select $\bf k$ such that $k_1=0$ and 
$$k_i-k_{i-1}= \begin{cases} i-2    , \   i\in\{ 2,\ldots ,N\}, i\leq \frac{N+2}{2}\\
N-i , \   i\in\{ 2,\ldots ,N \}, i> \frac{N+2}{2},\end{cases}$$
we do obtain an extended equilibrium $x^*$ that achieves the above bounds.
\end{proof}


This example shows that the system, provided $N$ is large enough, has extended equilibria that are arbitrarily far from the consensus. Figure~\ref{fig:line-convergence} shows convergence to a non-consensus state. 

\begin{figure}[htb]
\centering
\psfrag{xct}{$x$}
\psfrag{time}{$t$}
\includegraphics[width=0.6\columnwidth]{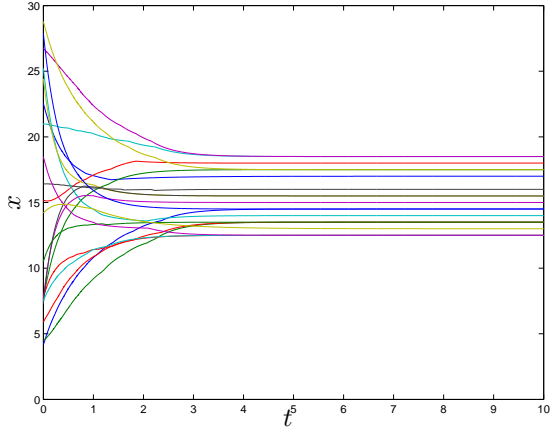}
\caption{Simulation via explicit Euler method of~\eqref{eq:line} for $N=20$ from random initial conditions in $[0,30]$, showing convergence to a non-consensus state.}
\label{fig:line-convergence}
\end{figure}

\section{Quantization as disturbance}\label{sec:quant-dist}

A classical approach to obtain (conservative) results about quantized systems is to see the quantized dynamics as the perturbation of a ``nominal'' system, where the perturbation is due to the quantization error $q(x)-x$. This idea is also useful for our problem. In this section we get an asymptotic estimate of the distance of solutions from consensus. In order to  prove it we need to assume the graph to be weight balanced. Below, we recall a few concepts of graph theory needed to make use of this assumption.

We have already defined  $d_i=\sum_{j=1}^N a_{ij}$. Let $d^i=\sum_{k=1}^N
a_{ki}$.  The (directed weighted) graph with adjacency matrix $A$ is said to be {\em weight-balanced} if  $d_i=d^i$ for all $i=1,...,N$. 
Note that
$L\1=0$ and $\1^\top L=0$ if and only if the graph is weight-balanced. 
Given an edge
$(i,j)$, we shall refer to $i$ and to $j$ as the tail and the head
of the edge, respectively. A path is an ordered list of
edges such that the head of each edge is equal to the tail of the
following one. The graph is said to be 
\begin{enumerate}\item {\em strongly connected} if
for any $i,j$ there is an path from $i$ to $j$;
\item {\em connected} if there exists one node $j$ such that for any $i$ there is path from $i$ to $j$;
\item {\em weakly connected} if for each pair
of nodes $i,j$, one can construct a path which connects $i$ and $j$ by possibly reverting the direction of some edges.
\end{enumerate}
These three notions of connectedness coincide when the graph is {\em symmetric}, that is when $a_{ij}>0$ if and only if $a_{ji}>0$. On the contrary, weakly connected weight-balanced graphs are strongly connected~\cite[Proposition~2]{JC:08}. We recall the following
result, which can be derived from~\cite[Theorem~1.37]{FB-JC-SM:09} and
\cite[Formula~(1)~and~Section~2.2]{JC:08}.

\begin{lemma}\label{lemma:t1}
Let $L$ be the Laplacian matrix of a weight-balanced and
weakly connected graph. Then:
\begin{itemize}
\item[(i)] The matrix ${\rm Sym}(L)=\frac{L+L^\top }{2}$ is positive semi-definite.
\item[(ii)] Denoted by $\lambda_*$ the smallest
non-zero eigenvalue of ${\rm Sym}(L)$,
$$
x^\top  L x \ge \lambda_*
\| x- x_a \1 \|  ^2\;, $$ for all $x\in \real^N$.
\end{itemize}
\end{lemma}

\begin{theorem}[Convergence to a set]\label{prop:convergence-to-M}  Assume that the graph with adjacency matrix $A$ is weight balanced and weakly connected.
If $x(t)$ is any Carath\'eodory or Krasowskii solution of~\eqref{eq:dynamics} and
$$M=\left\{ x\in\real ^N: \  \inf _{\alpha\in \real } \| x-\alpha \1
\|\leq \frac{||A||}{\lambda_*}\frac{\sqrt{N}}{2}\right\},
$$
then ${\dist}(x(t),M)\to 0$ as $t\to +\infty$.
\end{theorem}
\begin{proof} We prove the statement for Krasowskii solutions, as Carath\'eodory solutions are a special case.
Let $y(t)= x(t) -x_a(t)\1$. Then
$\dot y(t) = \dot x(t) - \dot x_a(t)\1 .$
Consider the function $V(y)=\frac{1}{2}y^\top y$. We have that 
\begin{align*}
\nabla V(y)\cdot \dot y
=& y^\top \dot y(t)\\
=& (x -x_a\1 )^\top [\dot x-\dot x_a\1]\\
=& (x -x_a\1 )^\top \dot x-x^\top \dot x_a\1+ x_a\1^\top \dot x_a\1\\
=& (x -x_a\1 )^\top \dot x-\dot x_ax^\top \1+ \dot x_aN  x_a\\
=& (x -x_a\1 )^\top \dot x-\dot x_aNx_a+ \dot x_aN  x_a\\
=& (x-x_a\1)^\top \dot x.
\end{align*}
From~\eqref{eq:dynamics-L} we get
$$\dot x \in  -L(x-x_a\1)+A\K(\Q (x)-x)\subseteq  -L(x-x_a\1)+A(\K\Q (x)-x).$$
For any $v\in \K\Q (x)-x$, it holds $\| v\|\leq \frac{\sqrt N}{2}$. Then, if $v\in \K\Q (x)-x$ is such that $\dot y=-L(x-x_a\1)+Av$, we have
\begin{align*}
\nabla V(y)\cdot \dot y
=& (x -x_a\1 )^\top [-L(x-x_a\1)+Av]\\
=& -(x -x_a\1 )^\top L(x-x_a\1)+(x -x_a\1 )^\top Av\\
=& -(x -x_a\1 )^\top {\rm Sym}(L) (x-x_a\1)+(x -x_a\1 )^\top Av\\
\leq& -\lambda_* \|x -x_a\1 \| ^2+ \|x -x_a\1 \| \| A\|\frac{\sqrt
N}{2}\\
\leq& \|x -x_a\1 \|\left[ -\lambda_* \|x -x_a\1 \| + \| A\|\frac{\sqrt
N}{2}\right].
\end{align*}
We conclude that ${\dist}(x(t),M)\to 0$ as $t\to +\infty$.
\end{proof}


Note that the infimum $\inf_{\alpha\in \real } \| x-\alpha \1\|$ is in fact achieved for $\alpha=\frac1N\sum_{i=1}^N x_i.$
Then we can further elaborate that 
$$M=\left\{ x\in\real ^N: \   \frac1{\sqrt{N}}\| x- \frac1N\sum_{i=1}^N x_i \1
\|\leq \frac{||A||}{2\lambda_*}\right\}.$$
 In order to better understand the interest of this overapproximation of the limit set, we specialize itt to two families of graphs with uniform weights: complete and path graphs with $a_{ij}=1$ if $i$ and $j$ are connected  and $a_{ij}=0$ if they are not.

\begin{example}[Complete graph] On complete graphs, we have 
$$\dot x_i=\sum_{j\not =i} \big(q(x_j)-x_i\big)\quad \text{for all $i$},$$ while $||A||=N-1$ and $\lambda_*=N$, so that asymptotically $$\frac1{\sqrt{N}}\| x- \frac1N\sum_{i=1}^N x_i \1
\|\leq \frac12.$$
This bound implies that the limit points are close to consensus. In fact, we shall prove below (Theorem~\ref{conv_compl}) that in this case limit points are precisely consensus points.
\end{example}

\begin{example}[Path graph]
On path graphs with $N$ nodes, $\lambda_*=1-\cos(\frac\pi{N})$  and $||A||\le 2$, because it always holds by Gershgorin's disk lemma that $||A||\le d_{\max}$. Then, the set $M$ is defined by 
\begin{align*}
\frac1{\sqrt{N}}\| x- \frac1N\sum_{i=1}^N x_i \1\|&
\leq \frac{2}{2 (1-\cos\frac\pi{N})}\\&=
\frac1{\frac{\pi^2}{2 N^2}-\frac{\pi^4}{4 N^4}+ o(\frac1{N^4})}
\\&=
\frac2{\pi^2} \frac{N^2 }{1-\frac{\pi^2}{2 N^2} + o(\frac1{N^2})}=\frac2{\pi^2} N^2 + O(1) \quad \text{as $N\to\infty$}.
\end{align*}
Let us now consider the following equilibrium $x^*\in E_e$ for odd $N=2m +1$, which was constructed in the proof of Proposition~\ref{prop:q-equil-path}: the point $x^*$ is symmetric with respect to the median value $x^*_{m+1}$ and such that
\begin{align*}
x^*_1&=0\\
x^*_i&=\frac{i^2-3i+3}2 \quad \text{if $2\le i\le m$}\\
x^*_{m+1}&=\frac{m(m-1)}2.
\end{align*}
For instance, for $N=9$ we have $\mathbf k=(0, 0, 1, 3, 6, 9, 11, 12, 12)^\top$ and correspondingly $x^*=(0,\frac12,\frac32,\frac72,6,\frac{17}2,\frac{21}2,\frac{23}2,12)^\top.$
Consequently, $x^*_a=x^*_{m+1}$ and
\begin{align*}
\| x^*-x^*_a\|^2=&
2\sum_{i=1}^{m} (x^*_i-x^*_{m+1})^2\\=& 
2\sum_{i=1}^{m} \left(\frac{i^2-3i +3}2-\frac{m(m-1)}2\right)^2\\=& 
\frac12\sum_{i=1}^{m} \big(i^4-6 i^3 - (2m^2-2m-15) i^2+ 6 (m^2-m-3) i \\&\qquad\qquad\qquad\qquad\qquad\qquad+ m^4-2m^3-5m^2+6m+9\big).
\end{align*}
By recalling that for all $a\in \natural$, $$\sum_{i=1}^n i^a= \frac1{a+1}{n^{a+1}}+ o(n^{a+1}) \quad \text{as $n\to\infty$}$$ and identifying the highest order terms in the above expression, we finally observe that 
\begin{align*}
\frac1{\sqrt{N}}\| x^*-x^*_a\|= \frac1{\sqrt{120}}\ N^2 + o(N^2) \quad \text{as $N\to\infty$}.
\end{align*}
\end{example}
This example shows that the result in Proposition~\ref{prop:convergence-to-M} can not be significantly improved in general, as the estimate on the limit set is asymptotically tight for large networks in the sense of the Euclidean distance from the consensus.
Nevertheless, stronger results can be obtained for specific topologies, as we do in the following section.

\section{Convergence to consensus on special graphs}\label{sec:consensus}
The previous sections emphasize the fact that in general we can not expect consensus for a system of the form~\eqref{eq:dynamics}. Nevertheless we can prove that consensus is achieved when the communication graph has some particular form, namely if it is complete or bipartite complete, and weights are uniform. 

\subsection{Complete graph}\label{sect:complete}

System~\eqref{eq:dynamics} in the case of the complete graph with uniform weights read 
\be
\label{eq:complete}
\dot x=f(x)
\ee
where $\map{f}{\real ^N}{\real ^N}$ is given by $f_i(x)=\sum_{j\not =i}[q(x_j)-x_i]$.
We define
$$
\begin{array}{rl} \q (x)& =(\sum_{j\not =1}q(x_j),\ldots , \sum_{j\not =N}q(x_j))^\top\\& =(\sum_{j =1}^Nq(x_j)-q(x_1),\ldots , \sum_{j=1}^Nq(x_j)-q(x_N) )^\top ,
\end{array} 
$$
and we observe that for every  ${\bf k}\in \integer^N$ the function $\q$ is constant on each set 
$S_{{\bf k}}$.
The discontinuous vector field can be written as 
\be
\label{fq} 
f(x)=\q(x)- (N-1)x.
\ee
This form makes evident  that in each set $S_{{\bf k}}$
trajectories are line segments.

We now prove that the set $E_e$ of extended equilibria reduces to quantized consensus points which actually  are Carath\'eodory equilibria as they belong to the interior of some $S_{\bf k}$.
\begin{proposition}[Extended equilibria for complete graph]
The sets of Carath\'eodory and extended equilibria of~\eqref{eq:complete} coincide and are equal to 
$$E_C=E_e=\{ x\in \integer ^N: \  \exists h\in \integer\text{ such that }\   x_i=h \  \forall i=1,\ldots ,N\}.$$ 
\end{proposition}
\smallskip

\begin{proof} Clearly, any point of $E_e$ is a Carath\'eodory equilibrium of~\eqref{eq:complete}.
On the other hand $x^*$ is an extended equilibrium if 
there exists ${\bf k}\in \integer ^N$  such that $f_{\bf k}(x^*)=0$ and  $x^*$ belongs to $ \overline{S_{{\bf k}}}$. If 
$K=\sum_{j=1}^Nk_j$, $x^*$ is a partial equilibrium if for all $i=1,\ldots,N$
$$x^*_i=\frac{K-k_i}{N-1}$$
and
$$k_i-\frac 12\leq x^*_i \leq k_i+\frac 12$$
By multiplying these bounds by $N-1$, we get 
$$(N-1)\left(  k_i-\frac 12\right)\leq K-k_i\leq (N-1)\left( k_i+\frac 12\right)$$
that implies, for all $i=1,\ldots,N$,
$$\frac{K}{N}-\frac{N-1}{2N}\leq k_i\leq \frac{K}{N}+\frac{N-1}{2N}.  $$
Since $k_i\in \integer $
interval  $\left[ \frac{K}{N}-\frac{N-1}{2N}, \frac{K}{N}+\frac{N-1}{2N}\right]$, and
 these bounds imply that $k_1=\ldots =k_n$. 
We conclude that if $x^*$ is an extended equilibrium then $x^*\in S_{(h,\ldots ,h)}$
for some $h\in \integer $. In this case $f(x^*)_i=(N-1)h-(N-1)x^*_i$
for all $i=1,\ldots ,N$, and finally $x^*_i=h$ for all $i=1,\ldots ,N$.
\end{proof}

We remark that also for~\eqref{eq:complete} there may be Krasowskii equilibria which are not extended equilibria: see, for instance, the initial condition in Example~\ref{example:non-unique}.
However, we will shortly see that all Krasowskii equilibria are consensus points.

\begin{remark}[Local stability]
Note that extended equilibria of~\eqref{eq:complete} are Lyapunov stable and locally asymptotically stable. In fact for any equilibrium
$x^*$ there exists a neighborhood where the  vector field can be written as
$f(x)=- (N-1)(x-x^*)$.
\end{remark}

\begin{remark}[Finite-time exit for Carath\'eodory solutions]\label{rem:finitetime}
From the proof of the previous proposition it
follows that in each $S_{{\bf k}}$ with ${\bf k}\not =(h,\ldots ,h)^\top $ for
any $h\in \integer$, trajectories are line segments that are generated by Carath\'eodory solutions corresponding  to a
vector field whose equilibrium is out of $S_{{\bf k}}$. This consideration implies that all solutions to~\eqref{eq:dynamics} escape from every set $S_{{\bf k}}$ with ${\bf k}\not =(h,\ldots ,h)^\top $ in finite time. 
On the other hand, once inside some $S_{(h,h,\ldots ,h)}$,
the equilibrium is reached asymptotically.
\end{remark}

We now state the main result of this section: any Carath\'eodory or Krasowskii solution converges to a consensus point. The convergence is illustrated in Figure~\ref{fig:complete-convergence-complete}.

\begin{figure}[htb]
\centering
\psfrag{xct}{$x$}
\psfrag{time}{$t$}
\hspace{-0.02\columnwidth}\includegraphics[width=0.53\columnwidth]{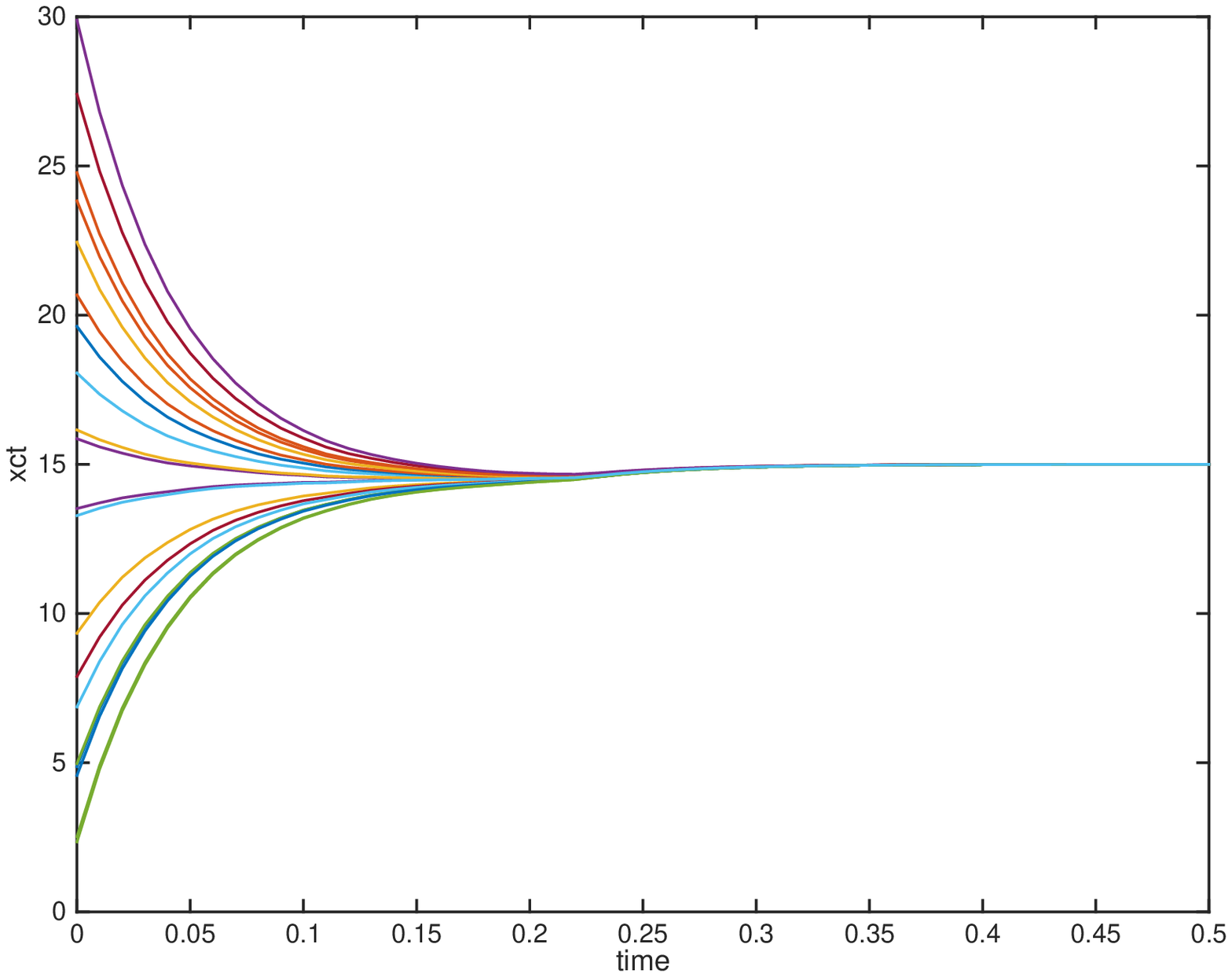}\hspace{-0.05\columnwidth}
\includegraphics[width=0.53\columnwidth]{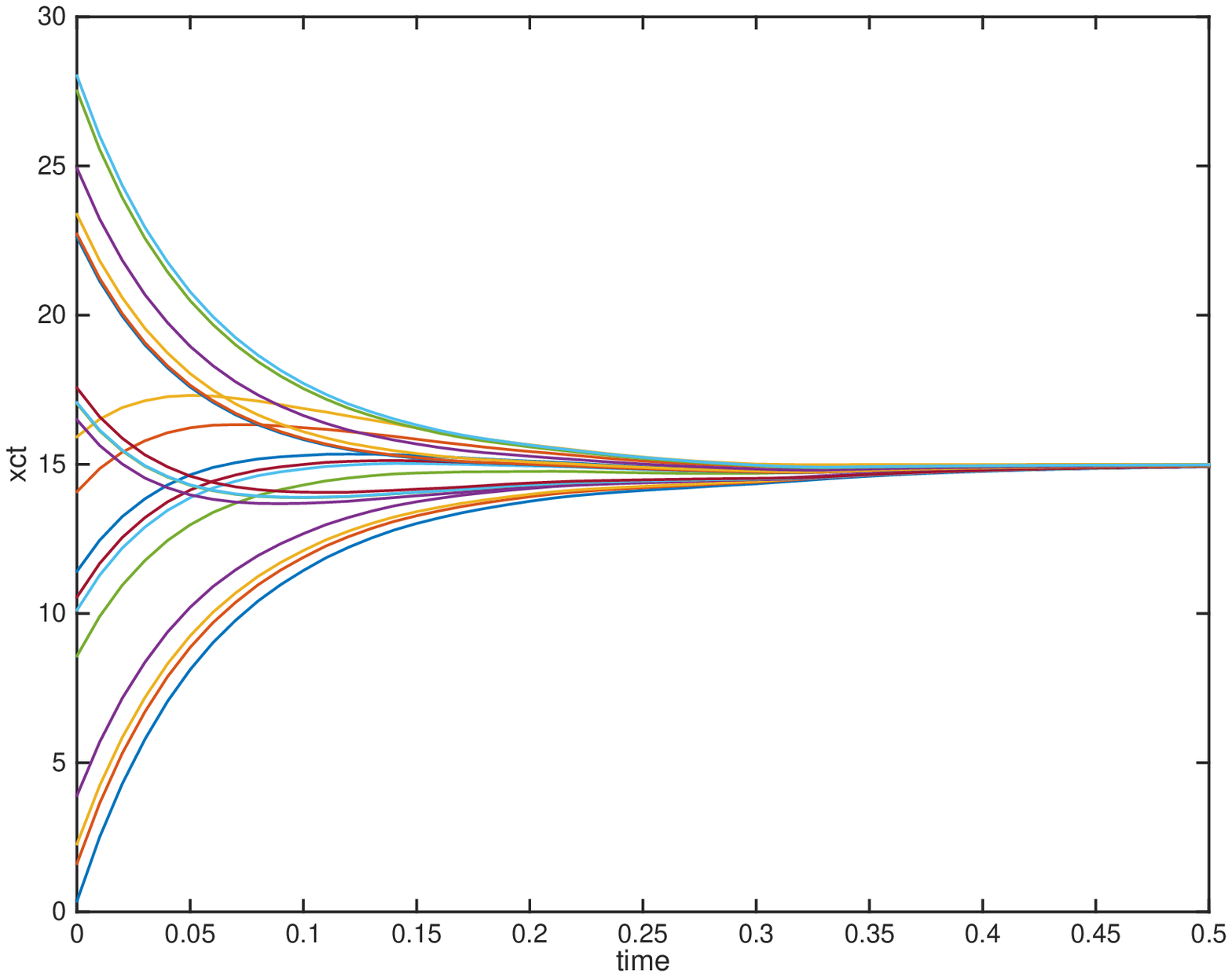}
\caption{Simulations of~\eqref{eq:complete} and~\eqref{eq:complete_bipartite} (respectively, left and right plot) for $N=20$ from random initial conditions in $[0,30]$, obtained via the explicit Euler method with step $10^{-2}$. Both simulations show convergence of the states to the integer consensus value $h=15$. System~\eqref{eq:complete} also exhibits order preservation as per Lemma~\ref{lemma:invariance}. }
\label{fig:complete-convergence-complete}
\end{figure}

\begin{theorem}[Convergence--Complete graph]\label{conv_compl}
Any Carath\'eodory or Krasowskii solution $x(t)$ of~\eqref{eq:complete} converges to a consensus point. Furthermore, if $x(t)$ is Carath\'eodory, then the limit point is necessarily of the form $(h,\ldots ,h)^\top $ with integer $h$. If instead $x(t)$ is Krasowskii, then the limit may be of the form $(h+\frac{1}{2},\ldots ,h+\frac{1}{2})^\top $.
\end{theorem}
This fact also implies that all Krasowskii equilibria of~\eqref{eq:complete} are consensus points.
In order to prove Theorem~\ref{conv_compl} we need the following lemma.

\begin{lemma}\label{lemma:invariance}
Let $x(t) $ be a Carath\'eodory or Krasowskii solution  of~\eqref{eq:complete} and let $s_{ij}(x)=x_i-x_j$ and $S_{ij}=\{ x\in \real ^N: s_{ij}(x)=0\}$. 
\begin{itemize}
\item[(A)] (Invariant manifolds) If $x (t_0) \in S_{ij}$, then $x(t) \in S_{ij}$
for all $t\geq t_0$.

\item[(B)] (Order preservation) If  $x_i(t_0)\leq x_j(t_0)$ for
  some $t_0\in \real$ then
  $x_i(t)\leq x_j(t)$  for all $t\geq t_0$.

\end{itemize}
\end{lemma}

\begin{proof}
We prove the statements for Krasowskii solutions. 
\noindent 
(A)
Let $x(t)$ be a Krasowskii solution such that $x(t_0)=x_0\in S_{ij}$. We prove that $x(t)$ can not leave $S_{ij}$.  If $x_0\in S_{ij}\backslash \Delta$, $f(x_0)$ is parallel to $S_{ij}$ in fact $\nabla s_{ij}(x_0)=e_i-e_j$
and 
\begin{align*}\nabla s_{ij}(x_0)\cdot f(x_0)=& f_i(x_0) -f_j(x_0)\\=
&\sum_{h\not =i}[q({x_0}_h)-{x_0}_i]-\sum_{h\not =j}[q({x_0}_h)-{x_0}_j]\\=
&q({x_0}_j)-q({x_0}_i)- (N-1)({x_0}_i-{x_0}_j)=0.
\end{align*}
Let now $x_0\in S_{ij}\cap \Delta$, i.e.\ ${x_0}_h=k_h+\frac{1}{2}$ for $h\in\I '\subset \I$ and some $k_h\in \integer$. We first consider the case $i,j\not \in \I '$. In this case for any limit value $f^l(x_0)$ of $f(x_0)$ at $x_0$ one has 
$$\nabla s_{ij}(x_0)\cdot f^l(x_0)= f^l_i(x_0) -f^l_j(x_0)= q({x_0}_j)-q({x_0}_i)- (N-1)({x_0}_i-{x_0}_j)=0.$$ 
meaning that all  $f^l(x_0)$ are parallel to $S_{ij}$. 

Then we consider the case ${x_0}_i={x_0}_j=k+\frac{1}{2}$ for some $k\in \integer$. We consider  four limit values   $f^{(i)}(x_0)$, $f^{(ii)}(x_0)$, 
$f^{(iii)}(x)$, $f^{(iv)}(x)$ of vector the field $f(x_0)$, depending on the fact that $x_0$ is approached from the following sectors:
\begin{itemize}
\item [(i)] $x_i>k+\frac{1}{2}$, $x_j<k+\frac{1}{2}$
\item [(ii)] $x_i>k+\frac{1}{2}$, $x_j>k+\frac{1}{2}$
\item [(iii)] $x_i<k+\frac{1}{2}$, $x_j>k+\frac{1}{2}$
\item [(iv)] $x_i<k+\frac{1}{2}$, $x_j<k+\frac{1}{2}$. 
\end{itemize}
Note that $$\nabla s_{ij}(x_0)\cdot f^{(ii)}(x_0)=\nabla s_{ij}(x_0)\cdot f^{(iv)}(x_0)=0,$$ 
meaning that  $f^{(ii)}(x_0)$ and $f^{(iv)}(x_0)$ are parallel to $S_{ij}$ at $x_0$.
In sector (i)  one has $s_{ij}(x)= x_i-x_j>0$ and 
$$\nabla s_{ij}(x_0)\cdot f^{(i)}(x_0)= -1.$$ 
This means that $f^{(i)}(x_0)$ points towards $S_{ij}$ and the solution can not move towards the sector (i).
Analogously in (iv)  one has $s_{ij}(x)= x_i-x_j<0$ and $\nabla s_{ij}(x_0)\cdot f^{(iv)}(x_0)= 1$. This means that $f^{(iv)}(x_0)$ points towards $S_{ij}$ and the solution can not move towards the sector (iv).

\noindent
(B) This is an immediate consequence of the previous statement. 
Assume first that $x_i(t_0)<x_j(t_0)$ and $x_i(t^*)>x_j(t^*)$ for some
$t^*>t_0$. Then there exist $\bar t\in (t_0,t^*)$ such that $x_i(\bar
t)=x_j(\bar t)$ and $t^*>t_0$ such that $x_i(t^*)>x_j(t^*)$, which
contradicts (A).
In the case $x_i(t_0)=x_j(t_0)$ and $x_i(t^*)>x_j(t^*)$ for some
$t^*>t_0$ the contradiction with (A) is immediate.
\end{proof}

\begin{proof}[Proof of Theorem~\ref{conv_compl}] As in the proof of the lemma we prove convergence for the more general case of Krasowskii solutions. Let $x(t)$ be any Krasowskii solution of~\eqref{eq:complete}.  
Thanks to the previous lemma, we can assume  without loss of generality that $x_1(t_0)\leq \dots\leq x_N(t_0)$. 
Indeed, by order preservation we have
$x_1(t)\leq \ldots \leq x_N(t)$ for all $t\geq t_0$ and $x_N(t)-x_1(t)\geq
0$ for all $t\geq t_0$.
We  prove that $x_N(t)-x_1(t)\to 0$ as $t\to +\infty $. 
If there exists $T$ such that $x_N(T)=x_1(T)$, then, thanks to Lemma~\ref{lemma:invariance}, $x_N(t)-x_1(t)= 0$ with $t\geq T$.
If instead $x_N(t)>x_1(t)$ for all $t\geq t_0$, then we can recall that $\dot x(t)\in -Dx(t)+A\K\Q (x(t))$, where $a_{ij}=1$ if $i\not =j$ and $a_{ij}=0$ if $i =j$, and there exists a measurable function $v(t)\in \K\Q (x(t))$  such that 
$\dot x(t)= -Dx(t)+Av(t)$.
We can then write
\begin{align*}
\frac{\dd}{\dd t}(x_N(t)-x_1(t))=&\sum _{j\not =N} v_j(t)- (N-1)x_N(t)- \sum _{j\not =1} v_j(t)+ (N-1)x_1(t)\\
=&- (N-1)(x_N(t)-x_1(t))- (v_N(t)-v_1(t))
\end{align*}

As $x_N(t)>x_1(t)$ we have that $v_N(t)\geq v_1(t)$, i.e.\ $v_N(t)-v_1(t)\geq 0$ and then
\begin{align*}
\frac{\dd}{\dd t}(x_N(t)-x_1(t))=&- (N-1)(x_N(t)-x_1(t))- (v_N(t)-v_1(t))\\ 
\leq &- (N-1)(x_N(t)-x_1(t)),
\end{align*}
which implies that $x_N(t)-x_1(t)\to 0$ as $t\to +\infty $. 

We now distinguish two cases. If $x(t)$ is in the interior of $S_{(h,\ldots ,h)}$ for some $t$
and some $h\in \integer$, then 
$$\dot x(t)=f(x(t))=- (N-1)(x(t)-(h,\ldots ,h)^\top )$$
 and $x(t)\to (h,\ldots ,h)^\top $ as $t\to +\infty$. 
Otherwise, $x(t) \to (h+\frac{1}{2},\ldots ,h+\frac{1}{2})^\top $ for some
$h\in \integer$. In this case, $x(t)$ actually reaches $(h+\frac{1}{2},\ldots ,h+\frac{1}{2})^\top$ in finite time. Afterwards, if $x(t)$ is a Carath\'eodory solution, it necessarily converges to $(h,\ldots ,h)$ or to $(h+1,\ldots ,h+1)$. If instead $x(t)$ is a Krasowskii solution, it may stay at the Krasowskii equilibrium $(h+\frac{1}{2},\ldots ,h+\frac{1}{2})$.
\end{proof}


\subsection{Complete bipartite graph}

In the case the graph is complete bipartite, equations~\eqref{eq:dynamics} read
\begin{align}\label{eq:complete_bipartite}
\dot x_i =& \sum _{h\in \cal Q} [q(x_h)-x_i]\\
\nonumber\dot x_h =& \sum _{i\in \cal P } [q(x_i)-x_h]
\end{align}
where ${\cal P}=\{ 1,\ldots ,p \}$ and ${\cal Q}=\{ p+1,\ldots ,N  \}$.
\begin{theorem}[Convergence--Complete bipartite]
Any Carath\'eodory or Krasowskii solution to~\eqref{eq:complete_bipartite} converges to a consensus point. 
\end{theorem}

\begin{proof}Let $x(t)$ be any Krasowskii solution of~\eqref{eq:complete_bipartite}. 
First of all we prove that  $x_i(t)-x_j(t)\to 0$ as $t\to +\infty$ for $i,j\in \cal P$ and  $x_h(t)-x_k(t)\to 0$ as $t\to +\infty$ for $h,k\in \cal Q$.
We recall that $x(t)$ is such that 
$$\dot x(t)\in -Dx(t) +\K(A\Q(x(t)))$$ 
and that there exists a measurable function $v(t)$ such that 
$\dot x(t)= v(t)- D x(t)$  with $v(t)\in \K(A\Q (x(t)))$.
In the case of the bipartite graph $d_{ii}=N-p$, $d_{hh}=p$, $(A\Q (x))_i=\sum _{h\in \cal Q} q(x_h)$,  $(A\Q (x))_h=\sum _{i\in \cal P} q(x_i)$ with $i=1,\ldots ,p$ and $h=p+1,\ldots ,N$.

We denote by $v_i(t)$ the $i$-th component of $v(t)$  and prove that $v_i(t)=v_j(t)$  for $i,j\in \cal P$.
Clearly if $x_h(t)\not = k+1/2$, $k\in \integer$, for all $h\in \cal Q$, then 
 $v_i(t)=\sum _{h\in \cal Q}q(x_h(t))$ for all $i\in \cal P$.
If  $x_h(t) = k_h+1/2$, for some  $k_h\in \integer$ and  for some $h\in \cal Q'\subset \cal Q$, there exist $2^{|\cal Q'|}=M$ limit values $Q^1,\ldots ,Q^M$ for $\Q(x)$ at $x(t)$, 
and 
$v(t)\in \K(A\Q (x(t))$ can be written
$v(t)=\alpha _1 AQ^1+\ldots \alpha _M AQ^M $,  with  $\alpha _1,\ldots ,\alpha _M\in [0,1]$ and   $\alpha _1+\ldots +\alpha _M=1$.
As $(A\Q(x))_i=\sum _{h\in \cal Q}q(x_h)$ for any $i=1,\ldots ,p$ we get that $v_i(t)=v_j(t)$ for any $i,j\in \cal P$.
Analogously, one has $v_h(t)=v_k(t)$  for $h,k\in \cal Q$.  

For any $i,j\in \cal P$ we thus get  
$$\frac{\dd}{\dd t} (x_i(t)-x_j(t))=v_i(t)-px_i(t)- v_j(t)+p x_j(t)=-p(x_i(t)-x_j(t))$$
which  implies that $x_i(t)-x_j(t)\to 0$ as $t\to +\infty$. Moreover if there exists $\tau $ such that $x_i(\tau )=x_j(\tau )$ for $i,j\in \cal P$, then $x_i(t)=x_j(t)$ for all $t\geq \tau$.
Analogously one has $x_h(t)-x_k(t)\to 0$ as $t\to +\infty$ for any $h,k\in \cal Q$ and if there exists $\tau ' $ such that $x_h(\tau ')=x_k(\tau ')$ for $i,j\in \cal P$, then $x_h(t)=x_k(t)$ for all $t\geq \tau '$.
The  fact that  $x_i(t)-x_j(t)\to 0$ and $x_h(t)-x_k(t)\to 0$ as $t\to +\infty$ also implies that either there exists $T$ such that $x_i(t)\leq x_h(t)$ for all $i\in \cal P$, $h\in \cal Q$  and for all $t\geq T$ or $x_i(t)-x_h(t)\to 0$ for all $i\in \cal P$, $h\in \cal Q$.

Let $m(t)\in\{ 1,\ldots ,N\} $ be such that $x_{m(t)}(t)=\min \{ x_i(t), i=1,\ldots ,N\} $ and $M(t)\in\{ 1,\ldots ,N\} $ be such that $x_{M(t)}(t)=\max \{ x_i(t), i=1,\ldots ,N\} $.
We will simply denote $x_{m(t)}(t)=x_m(t)$ and $x_{M(t)}(t)=x_M(t)$. Thanks to  the previous remark we can assume that  $m(t)\in \cal P$ and $M(t)\in \cal Q$ definitively.
We first assume that there exist $T'\in \real $ and $q\in \integer $   such that $x_m(T'), x_M(T')\in (q-1/2,q+1/2)$. 
In this case, from~\eqref{eq:complete_bipartite} we get that $x_m(t), x_M(t)\in (q-1/2,q+1/2)$ for all $t\geq T'$  and $x_i(t)\to q$ for all $i\in \cal P$, $x_h(t)\to q$ for all $h\in \cal Q$.

We then prove that there exist $T$ and $q\in \integer $ such that either  $x_m(T), x_M(T)\in (q-1/2,q+1/2)$ or  $x_m(t)= x_M(t)=q+1/2$ for all $t\geq T$.

Note that as far as $x_m(t)\in [q-1/2, q+1/2)$ and $x_M(t)>q+1/2$,  for almost all $t$ one has
$$\frac{\dd}{\dd t}(x_M(t)-x_m(t))\in \K\Big(\sum _{i\in \cal P } [q(x_i)-x_M]-\sum _{h\in \cal Q } [q(x_h)-x_m]\Big)$$
and if $v\in  \K(\sum _{i\in \cal P } [q(x_i)-x_M]-\sum _{h\in \cal Q } [q(x_h)-x_m])$, then $v\leq -1$. This fact implies that 
$x_M(t)-x_m(t)$ decreases until one of the following cases is reached:

\begin{itemize}
\item{} $q-1/2<x_m(t')<q+1/2,\  x_M(t')=q+1/2$,
\item{} $x_m(t')=q-1/2,\  x_M(t')=q+1/2$,
\item{} $x_m(t')=x_M(t')=q+1/2$.
\end{itemize}

In the first case, $q(x_i(t'))=q$.  Since  $q-(q+1/2)=-1/2<0$, we get  that there exists $t''>t$ such that $x_M(t)$ decreases in $(t',t'')$ and  $x_m(T), x_M(T)\in (q-1/2,q+1/2)$ at some time $T\in (t',t'')$.

The second case is analogous to the first one with  $x_m(t)$ increasing and $x_M(t)$ decreasing in a right  neighbourhood $(t', t'')$ of $t'$, so that $x_m(T), x_M(T)\in (q-1/2,q+1/2)$ at some time $T\in (t',t'')$.

Finally in the third case it is sufficient to recall that if $x_i(t')=x_j(t')$  and $x_h(t')=x_k(t')$ then $x_i(t)=x_j(t)$  and $x_h(t)=x_k(t)$ for all $t\geq t'$. This implies that  $x_i(t), x_h(t)\to q$  or  $x_i(t), x_h(t)\to q+1$ or $x_i(t)= x_h(t)= q+1/2$ for all $i\in \cal P$, $h\in \cal Q$.

By exhausting the cases, we have proved the claim and consequently completed the proof.
\end{proof}


\section{Simulations on general graphs}\label{sec:simulations}
In order to explore the convergence properties of~\eqref{eq:dynamics}, we have constructed numerical solutions by applying the explicit Euler method from random initial conditions. These simulations show that non-consensus equilibria are quite common on non-structured graphs: two examples are shown in Figures~\ref{fig:geometric-graph} and~\ref{fig:random-graph}. Instead, we are unable to observe in the simulations any chattering phenomena that would indicate attractive Krasowskii solutions sliding on the discontinuity set.

\begin{figure}[htb]
\centering
\psfrag{xct}{$x$}
\psfrag{q(xct)}{$q(x)$}
\psfrag{time}{$t$}
\includegraphics[width=0.46\columnwidth]{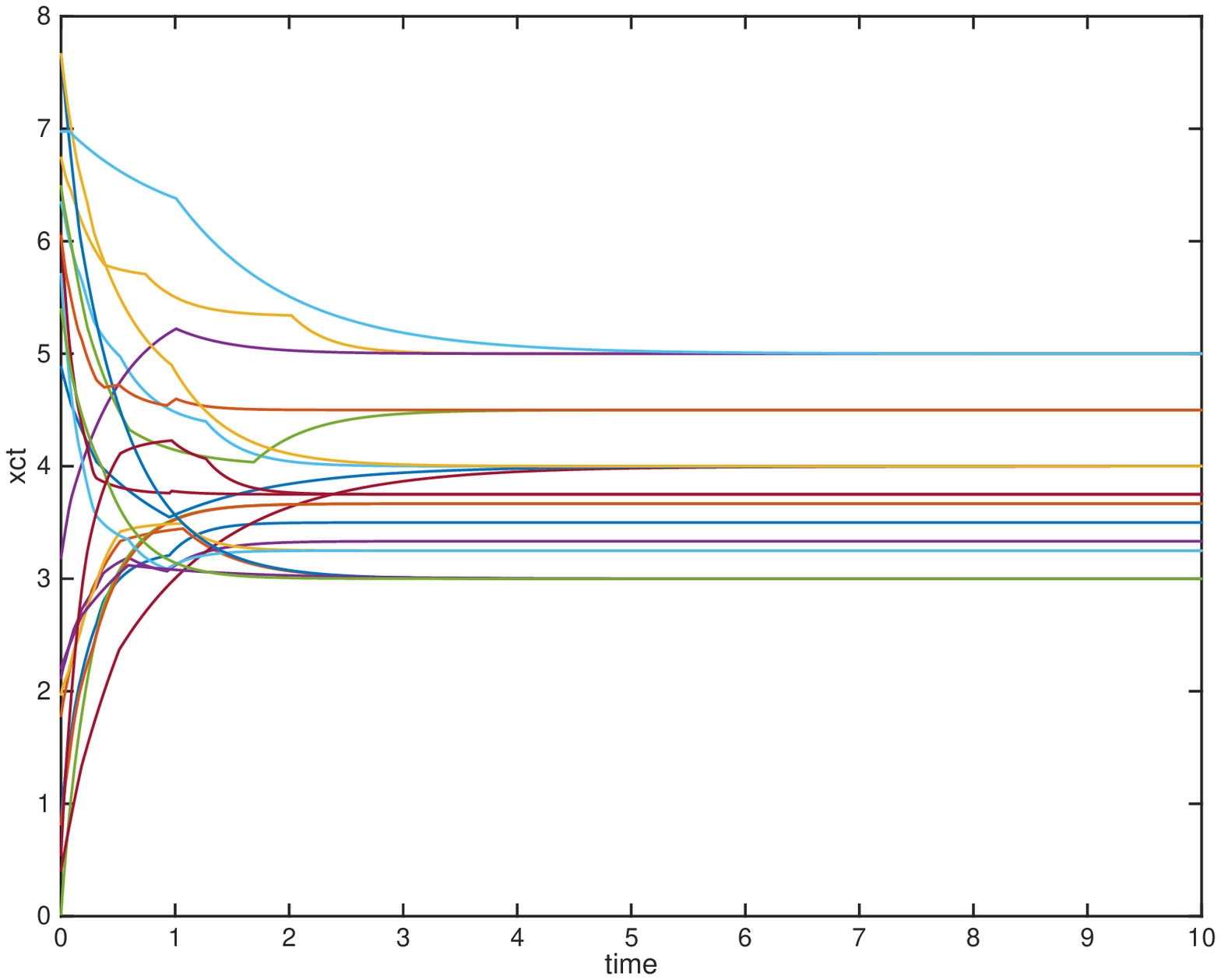}\hspace{0.05\columnwidth}
\includegraphics[width=0.46\columnwidth]{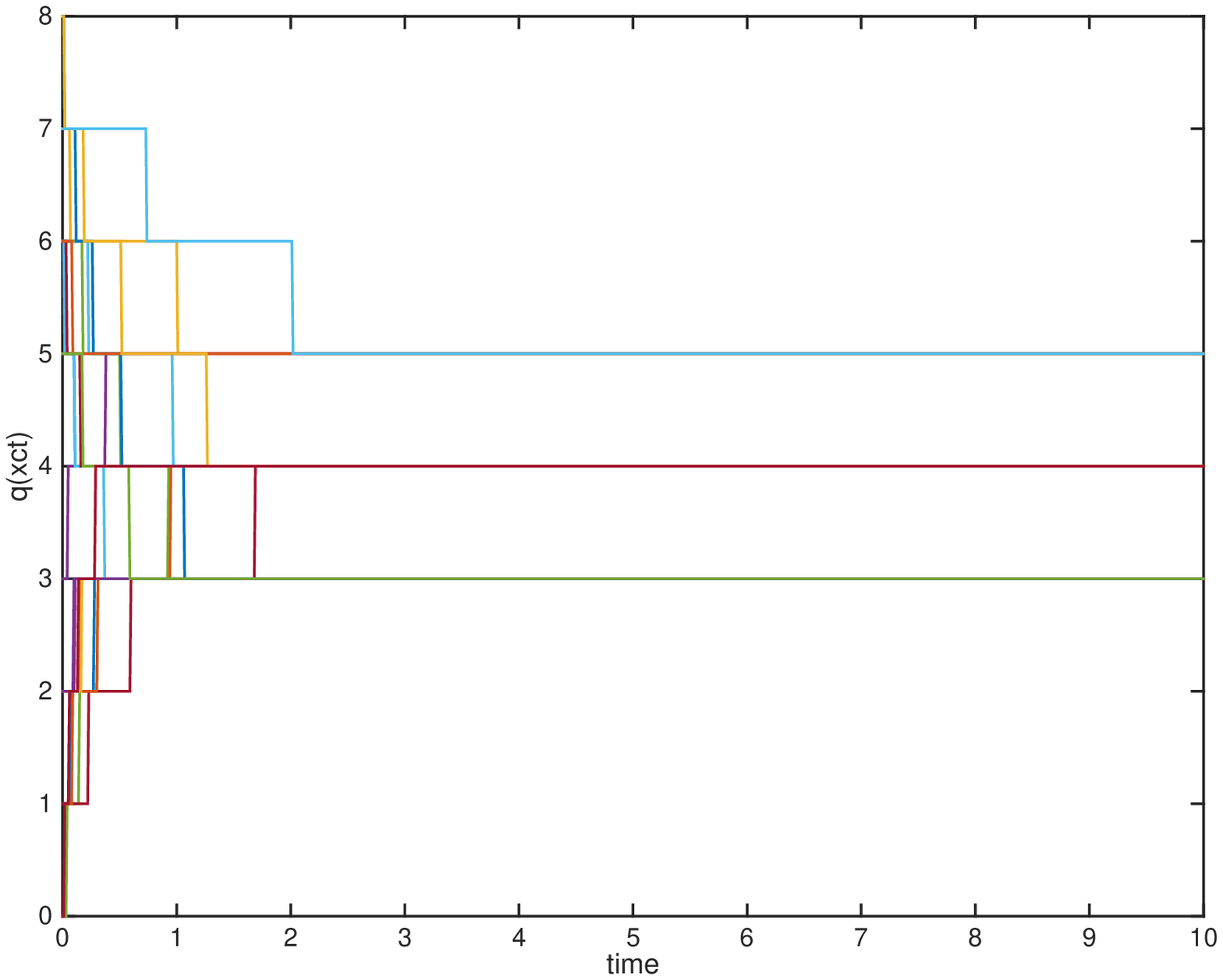}
\caption{Evolution of $x$ and $q(x(t))$ on a connected graph with $N=20$ generated as an instance of random geometric graph with connectivity radius $0.2$. The Euler method step is $0.01$.}
\label{fig:geometric-graph}
\end{figure}

\begin{figure}[htb]
\centering
\psfrag{xct}{$x$}
\psfrag{q(xct)}{$q(x)$}
\psfrag{time}{$t$}
\includegraphics[width=0.46\columnwidth]{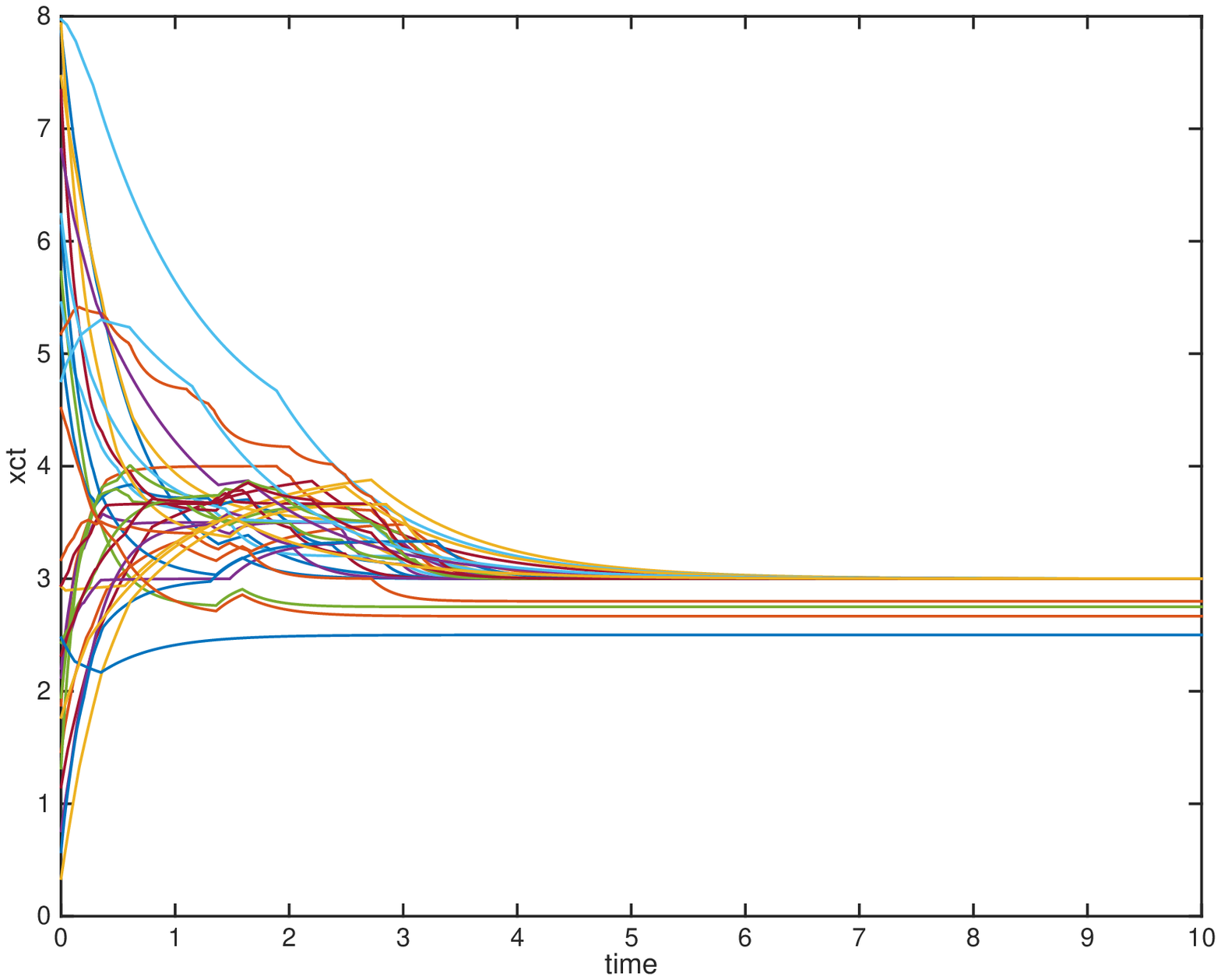}\hspace{0.05\columnwidth}
\includegraphics[width=0.46\columnwidth]{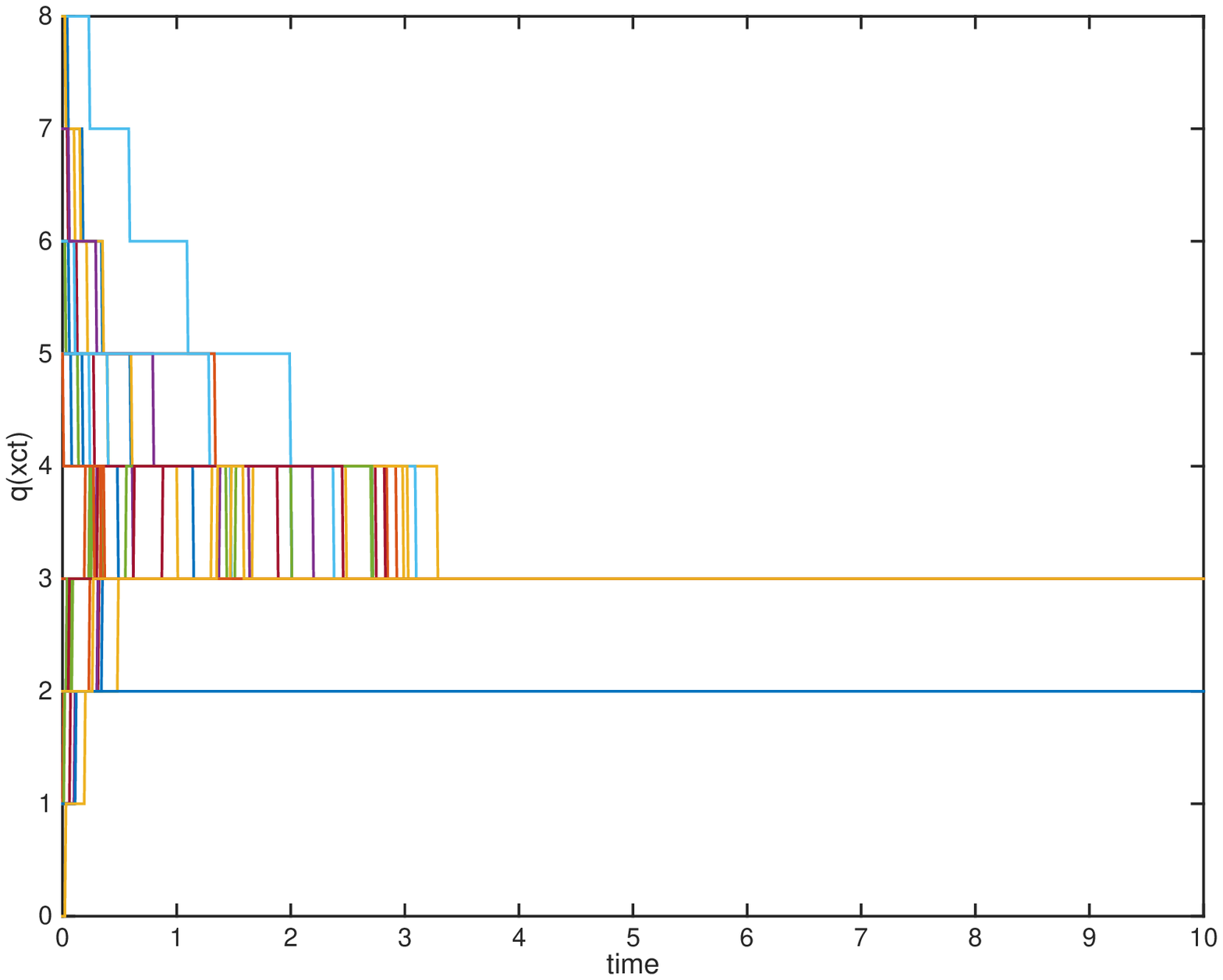}
\caption{Evolution of $x$ and $q(x(t))$ on a connected directed graph with $N=20$, randomly generated by connecting each ordered pair of nodes with probability $0.1$. The Euler method step is $0.01$.}
\label{fig:random-graph}
\end{figure}

\section{Conclusion}\label{sec:outro}
In this article we have addressed the problem of opinions formation in a group where opinions are continuous variables that are communicated to other individuals by means of discrete variables, which may represent behaviours, choices, or limited verbalization capabilities. 
We have shown that in general we can not expect consensus and unexpected equilibria may appear. We think this is interesting from the point of view of opinion dynamics, as a good model should be able to justify agreement and disagreement at the same time. Our positive results consist in an asymptotic estimate of the  distance from consensus depending on the communication graph and in the proof of consensus for the special cases the communication graph is either complete or bipartite complete. Finding a necessary and sufficient condition for consensus remains an open problem.

\bibliographystyle{plainurl}

\end{document}